\documentclass[11pt]{article}
\usepackage{amsmath,amssymb,amsthm}
\usepackage{mathtools}
\usepackage{geometry}
\usepackage{hyperref}
\usepackage{enumitem}
\theoremstyle{plain}
\newtheorem{lemma}{Lemma}
\newtheorem{theorem}{Theorem}
\newtheorem{corollary}{Corollary}
\theoremstyle{definition}
\newtheorem{definition}{Definition}
\newtheorem{remark}{Remark}
\newcommand{\F}{\mathbb{F}}
\newcommand{\supp}{\mathrm{supp}}
\newcommand{\wt}{\mathrm{wt}}

\title{\textbf{A Symplectic Proof of the Quantum Singleton Bound}}
\author{Frederick Dehmel\thanks{Department of Electrical Engineering and Computer Sciences, University of California, Berkeley, Berkeley, CA 94720, USA. Email: \texttt{dehmelf@berkeley.edu}} \qquad Shilun Li\thanks{Department of Mathematics, University of California, Berkeley, Berkeley, CA 94709, USA. Email: \texttt{shilun@berkeley.edu}}}
\date{}

\begin{document}
\maketitle

\begin{abstract}
We give a self-contained symplectic linear-algebraic derivation of the Quantum Singleton Bound
for stabiliser quantum error-correcting codes, and---our primary contribution---a complete
\texttt{Lean4} formalisation of that derivation.  Working in the language of finite-dimensional
symplectic vector spaces modelling Pauli operators, we assemble three ingredients that are
standard in the stabiliser formalism---distance-based erasure correctability, the cleaning
lemma of Bravyi and Terhal, and a dimension count---into a purely algebraic proof of the bound
$k + 2(d-1) \le n$ for any $[[n,k,d]]$ stabiliser code.  The argument invokes neither
von~Neumann entropy nor the no-cloning theorem, which makes it particularly well-suited to
mechanisation: built on \texttt{Mathlib}, the formalisation is, to our knowledge, the first
machine-checked proof of the Quantum Singleton Bound.  Our aim is thus not a new bound but a
streamlined, entropy-free presentation of a known argument together with its formal
verification.
\end{abstract}
\newpage
\tableofcontents
\newpage
\section{Introduction}\label{sec:introduction}

Quantum error-correcting codes are essential for protecting quantum information against
decoherence and operational noise, a prerequisite for scalable fault-tolerant quantum
computation~\cite{KnillLaflamme1997,Gottesman1997,CalderbankShor1996,Steane1996}.  Among the
most fundamental constraints on such codes is the \emph{Quantum Singleton Bound}: for an
$[[n,k,d]]$ stabiliser code,
\[
  k + 2(d-1) \le n.
\]
This is the quantum analogue of the classical Singleton bound $k + d - 1 \le n$ for linear codes
over finite fields~\cite{Singleton1964,MacWilliamsSloane1977}, and it places a hard limit on
the trade-off between the number of encoded logical qubits~$k$, the code distance~$d$, and
the block length~$n$.  Codes that achieve equality are called \emph{quantum maximum-distance
separable} (MDS) codes and have been the subject of extensive
study~\cite{GrasslBeth1999,RainsMDS}.

The bound was first established by Knill and Laflamme~\cite{KnillLaflamme1997} and
independently by Bennett, DiVincenzo, Smolin, and
Wootters~\cite{BennettDiVincenzoSmolinWootters1996} in the context of general quantum
error-correcting codes.  Their arguments rely on properties of quantum states and channels,
and the standard textbook proof~\cite{NielsenChuang} proceeds via von~Neumann entropy
inequalities: one shows that if a code can correct erasures of $d-1$ qubits, then the encoded
information must be recoverable from the remaining $n-(d-1)$ qubits, and an entropy argument
constrains how much information can be stored.  More recently, Grassl, Huber, and
Winter~\cite{GrasslHuberWinter2022} gave a streamlined entropic proof and extended the
technique to entanglement-assisted and catalytic codes.

For \emph{stabiliser} codes, however, the underlying mathematical structure is fundamentally
algebraic, not information-theoretic.  The $n$-qudit Pauli group (modulo phases) is naturally
identified with a $2n$-dimensional vector space over a prime field~$\F_p$, equipped with a
symplectic bilinear form that encodes commutation
relations~\cite{CalderbankRainsShorSloane1997,CalderbankRainsShorSloane1998,Gottesman1997}.
A stabiliser code is then specified by an isotropic subspace of this symplectic space, and
concepts such as error detection, correction, and code distance admit clean algebraic
formulations.  A natural question is therefore whether the Quantum Singleton Bound can be
derived purely within this symplectic framework, without invoking entropy or the theory of
quantum channels.

In this paper we give such a proof.  The argument combines three ingredients, each of which is
standard in the stabiliser literature but which together yield a self-contained and concise
derivation:
\begin{enumerate}[label=(\roman*)]
  \item \textbf{Distance implies erasure correctability.}  If a stabiliser code has distance
        $d$, then any set of at most $d-1$ qudit positions forms a correctable erasure, which
        in the symplectic model is captured by the inclusion $S^\perp \cap V_E \subseteq S$.
  \item \textbf{The cleaning lemma.}  For any partition of qudit positions into a set $M$ and
        its complement $M^c$, the dimensions of the spaces of ``supportable'' logical operators
        on $M$ and $M^c$ satisfy $g(M) + g(M^c) = 2k$.  This complementarity is the
        dimension-counting form of the cleaning lemma of Bravyi and
        Terhal~\cite{BravyiTerhal2009}, as presented in Preskill's lecture
        notes~\cite{Preskill2022TopologicalStabilizerCodes}; it captures how logical
        information can be ``cleaned'' off a correctable region.
  \item \textbf{A dimension argument.}  Given two disjoint correctable sets $A$ and $B$ with
        $|A|=|B|=d-1$, the cleaning lemma forces all $2k$ independent logical operators to be
        representable on the complement $C = [n] \setminus (A \cup B)$, yielding
        $k \le |C| = n - 2(d-1)$.
\end{enumerate}

The resulting proof is short, elementary, and entirely algebraic.  It avoids
the use of von~Neumann entropy, the no-cloning theorem, and the decoupling
principle---tools that, while powerful, are not intrinsic to the stabiliser formalism.

\paragraph{Formal verification.}
A secondary contribution of this work is a complete formalisation of the above argument in the
\texttt{Lean4} proof assistant, using the \texttt{Mathlib} library for finite-dimensional
linear algebra.  The formalisation covers the symplectic space model, the definition of
stabiliser codes, the cleaning lemma, and the final bound, totalling approximately 1300 lines
of \texttt{Lean4} code.  This is, to our knowledge, the first machine-checked proof of the
Quantum Singleton Bound.

Formal verification of results in quantum information theory remains relatively rare.  Most
existing work in this area has focused on the verification of quantum programs and circuits
rather than on coding-theoretic bounds.  Tools such as
SQIR~\cite{SQIR2021}, QHLProver~\cite{QHLProver2019}, CoqQ~\cite{CoqQ2023}, and the recent
Veri-QEC framework~\cite{VeriQEC2025} target operational verification---checking that a given
circuit implements the intended unitary or that a specific error-correction procedure is
correct.  By contrast, our formalisation concerns a structural impossibility result about the
parameters of any stabiliser code, which is a qualitatively different kind of claim.  The
linear-algebraic character of the proof makes it a natural candidate for mechanisation:
the key steps are subspace dimension calculations and inclusions, which interact well with the
finite-dimensional linear algebra API available in \texttt{Mathlib}.  An entropy-based proof
would require formalising the von~Neumann entropy and its subadditivity properties, a
substantially heavier undertaking.

\section{Related work}\label{sec:related}

\paragraph{Proofs of the Quantum Singleton Bound.}
The Quantum Singleton Bound was first proved by Knill and
Laflamme~\cite{KnillLaflamme1997}, who established the general theory of quantum
error-correcting codes and the conditions under which errors can be corrected.  The bound also
appears in work of Rains~\cite{Rains1999}, who proved it (along with other bounds) using
quantum weight enumerators and a linear programming approach.  The textbook treatment in
Nielsen and Chuang~\cite{NielsenChuang} derives the bound from the no-cloning theorem and
entropy considerations.  More recently, Grassl, Huber, and
Winter~\cite{GrasslHuberWinter2022} provided a clean entropic proof and extended it to
entanglement-assisted quantum error-correcting codes (EAQECCs) and catalytic codes (CQECCs),
correcting a previously conjectured generalisation that turned out to be
false~\cite{GrasslCounterexample2021}.  Klappenecker and
Sarvepalli~\cite{KlappeneckerSarvepalli2008} proved that no $\F_q$-linear subsystem code over
a prime field can beat the Quantum Singleton Bound.

\paragraph{The stabiliser formalism and symplectic geometry.}
The symplectic vector space model of the Pauli group was introduced by Calderbank, Rains,
Shor, and Sloane~\cite{CalderbankRainsShorSloane1997,CalderbankRainsShorSloane1998} and
independently by Gottesman~\cite{Gottesman1997}.  In this framework, an $n$-qudit stabiliser
code over $\F_p$ is described by an isotropic subspace of $\F_p^{2n}$ with respect to the
standard symplectic form, and code properties such as distance and error correction translate
into statements about subspace intersections and dimensions.  The framework extends naturally
to nonbinary (qudit) codes~\cite{Ashikhmin2001,KetkarKlappeneckerKumarSarvepalli2006} and
has been used to study the structure of logical operators~\cite{Wilde2009} and the synthesis
of Clifford circuits via symplectic transformations~\cite{Rengaswamy2018}.

\paragraph{The cleaning lemma.}
The cleaning lemma, which asserts that logical operators on a correctable region can be
``cleaned off'' to the complementary region, was first introduced by Bravyi and
Terhal~\cite{BravyiTerhal2009} in their no-go theorem for two-dimensional self-correcting
stabiliser memories, and was combined with entropic arguments by Bravyi, Poulin, and
Terhal~\cite{BravyiPoulinTerhal2010} to derive locality-constrained tradeoffs.  It has since
become a central tool in the study of how geometric locality and connectivity constrain code
parameters~\cite{BaspinGuruswamiKrishnaLi2023,DaiLi2024}.  A dimension-counting form is
presented in Preskill's lecture notes~\cite{Preskill2022TopologicalStabilizerCodes}, and
Kalachev and Sadov~\cite{KalachevSadov2022} gave a systematic linear-algebraic treatment,
showing that the cleaning lemma is at its core a fact about inner product spaces that admits a
lattice-theoretic abstraction.  The version we use here follows the dimension-counting
formulation, cast in symplectic language.

\paragraph{Formal verification of quantum information.}
Formal verification of quantum computing results has seen growing interest.  Hietala et
al.~\cite{SQIR2021} developed SQIR, a small quantum intermediate representation formalised in
Coq, and used it to verify properties of quantum circuits including Grover's algorithm.  Liu
et al.~\cite{QHLProver2019} formalised quantum Hoare logic in Isabelle/HOL.  Zhou et
al.~\cite{CoqQ2023} built CoqQ, a foundational Coq framework for verifying quantum programs
based on the quantum-while language.  Most recently, Huang et
al.~\cite{VeriQEC2025} proposed Veri-QEC, which formalises a program logic for quantum error
correction in Coq and provides automated verification of fault-tolerant protocols.  These
efforts all target the verification of \emph{quantum programs} or \emph{circuits}.  Our
contribution is complementary: we formalise a \emph{coding-theoretic bound}, which is a
statement about the parameter space of all possible stabiliser codes rather than about any
specific code or circuit.

\section{Contributions}\label{sec:contributions}

Concretely, our contributions are as follows:
\begin{enumerate}[label=(\roman*)]
  \item \emph{(Primary contribution.)}  We provide, to our knowledge, the first machine-checked
        proof of the Quantum Singleton Bound: a complete \texttt{Lean4} formalisation, built on
        \texttt{Mathlib}, covering the full linear-algebraic development---the symplectic form,
        stabiliser codes, erasure correctability, the cleaning dimension identity, and the main
        theorem.
  \item In support of the formalisation, we isolate a self-contained symplectic
        linear-algebraic derivation of the bound $k + 2(d-1) \le n$ from distance-based erasure
        correctability and the cleaning lemma via a dimension argument, using neither entropy
        nor the no-cloning theorem.  The individual ingredients are standard in the stabiliser
        literature; the contribution here is a streamlined, entropy-free and
        formalisation-ready presentation rather than a new bound.
\end{enumerate}

\section{Conventions and symplectic model}\label{sec:symplectic}

\begin{remark}[Field convention]
Throughout, $\F$ denotes a prime field $\F_p$.  We work in the symplectic vector space
\[
  V := \F^n \times \F^n \cong \F^{2n}
\]
with elements written $v = (v_X, v_Z)$.  In the stabiliser formalism, a vector
$v = (v_X, v_Z) \in V$ represents (the phase-free part of) the $n$-qudit Pauli operator
$X^{v_X} Z^{v_Z} = \prod_{i=1}^n X_i^{v_X(i)} Z_i^{v_Z(i)}$, where $X_i$ and $Z_i$ are the
generalised Pauli operators on the $i$-th qudit.  Two Pauli operators commute if and only if
the symplectic pairing of their corresponding vectors vanishes.
\end{remark}

\begin{definition}[Standard symplectic form]
For $u = (u_X, u_Z)$ and $v = (v_X, v_Z)$ in $V$, define
\[
  \langle u, v \rangle := \sum_{i=1}^n \bigl( u_X(i) v_Z(i) - u_Z(i) v_X(i) \bigr) \in \F.
\]
This is a bilinear, alternating (and hence antisymmetric) form on $V$.
\end{definition}

\begin{lemma}[Nondegeneracy]\label{lem:nondeg}
If $u \in V$ satisfies $\langle u, v \rangle = 0$ for all $v \in V$, then $u = 0$.
\end{lemma}

\begin{proof}
Fix an index $i \in [n]$.  Pairing $u$ with $(0, e_i)$ gives $u_X(i) = 0$, and pairing $u$
with $(e_i, 0)$ gives $-u_Z(i) = 0$.  Since $i$ was arbitrary, all coordinates of $u$ vanish.
\end{proof}

Nondegeneracy ensures that the symplectic form identifies $V$ with its dual, which is the
basis for the dimension identities used throughout the proof.

\section{Support, restriction, and distance}\label{sec:support}

Let $[n] := \{1, \dots, n\}$.

\begin{definition}[Support and weight]
For $v \in V$, define
\[
  \supp(v) := \{ i \in [n] : (v_X(i), v_Z(i)) \neq (0,0) \},
  \qquad
  \wt(v) := |\supp(v)|.
\]
Intuitively, $\supp(v)$ is the set of qudit positions on which the corresponding Pauli
operator acts nontrivially.
\end{definition}

\begin{definition}[Support subspace]
For $C \subseteq [n]$, define
\[
  V_C := \{ v \in V : v_X(i) = v_Z(i) = 0 \text{ for all } i \notin C \}.
\]
This is the subspace of $V$ consisting of vectors whose support is contained in $C$.
\end{definition}

\begin{lemma}[Dimension of a support subspace]\label{lem:dimVC}
$\dim(V_C) = 2|C|$.
\end{lemma}

\begin{proof}
Restriction to the coordinates indexed by $C$ gives a vector space isomorphism
$V_C \cong \F^{|C|} \times \F^{|C|}$.
\end{proof}

\begin{definition}[Restriction map]
For $E \subseteq [n]$, define $r_E : V \to V_E$ by zeroing all coordinates outside $E$:
$(r_E(v))_X(i) = v_X(i)$ and $(r_E(v))_Z(i) = v_Z(i)$ for $i \in E$, and zero otherwise.
\end{definition}

\begin{definition}[Stabiliser subspace and distance]\label{def:stabcode}
A stabiliser code $Q$ is specified by an isotropic subspace $S \subseteq V$
(i.e., $S \subseteq S^\perp$) with $\dim(S) = n - k$.  The (Pauli) distance is
\[
  d := \min\{ \wt(v) : v \in S^\perp \setminus S,\; v \neq 0 \}.
\]
Elements of $S$ represent stabiliser operators, and elements of $S^\perp \setminus S$
represent nontrivial logical operators.  The quotient $S^\perp / S$ is the logical operator
space.
\end{definition}

\begin{lemma}[Logical space dimension]\label{lem:logdim}
If $\dim(S) = n - k$, then $\dim(S^\perp / S) = 2k$.
\end{lemma}

\begin{proof}
Nondegeneracy of the symplectic form gives $\dim(S) + \dim(S^\perp) = 2n$, so
$\dim(S^\perp) = n + k$.  Therefore $\dim(S^\perp / S) = (n + k) - (n - k) = 2k$.
\end{proof}

\section{Erasure correctability}\label{sec:erasure}

\begin{definition}[Correctable erasure (commutant form)]\label{def:correctable}
A subset $E \subseteq [n]$ is correctable if
\[
  S^\perp \cap V_E \subseteq S.
\]
In the stabiliser setting, this condition is equivalent to the standard Knill--Laflamme
error-correction conditions for the erasure channel that traces out the qudits in
$E$~\cite{KnillLaflamme1997,NielsenChuang}.  Concretely, $E$ is correctable if and only if no
nontrivial logical operator is supported entirely within $E$: every element of $S^\perp$ that
is supported on $E$ must already lie in the stabiliser $S$.
\end{definition}

\begin{lemma}[Distance implies erasure correctability]\label{lem:dist-correct}
If $|E| \le d - 1$, then $E$ is correctable.
\end{lemma}

\begin{proof}
Suppose $v \in S^\perp \cap V_E$ with $v \notin S$.  Then $v \in S^\perp \setminus S$ and
$\wt(v) \le |E| \le d - 1$, contradicting the definition of~$d$.
\end{proof}

\section{Cleaning lemma and the dimension identity}\label{sec:cleaning}

The cleaning lemma captures a remarkable complementarity property of stabiliser codes: logical
operators that ``live'' on a subset $M$ of qudit positions and those that live on the
complementary set $M^c$ together account for all $2k$ independent logical degrees of freedom.

\begin{definition}[Supportable logical operators]
For $M \subseteq [n]$, define
\[
  g(M) := \dim\bigl( (S^\perp \cap V_M) / (S \cap V_M) \bigr).
\]
This counts the number of independent logical Pauli operators (modulo stabilisers) that can be
supported entirely on $M$.
\end{definition}

\begin{lemma}[Cleaning dimension identity]\label{lem:cleaning}
For any $M \subseteq [n]$,
\[
  g(M) + g(M^c) = 2k.
\]
\end{lemma}

\begin{proof}
The following argument is a linear-algebraic reformulation of the dimension-counting proof
appearing in Preskill's notes~\cite{Preskill2022TopologicalStabilizerCodes}.

Let $P \cong V$ denote the abelianised Pauli group viewed as a $2n$-dimensional $\F$-vector
space with the same symplectic form.  Fix $M \subseteq [n]$ and write $M^c$ for its
complement.  Let $P_M := V_M$ and $P_{M^c} := V_{M^c}$.

Define the stabiliser-restriction subspaces
\[
  S_M := S \cap P_M, \qquad S_{M^c} := S \cap P_{M^c}.
\]
Choose a complementary subspace $S_0$ such that
\[
  S = S_M \oplus S_{M^c} \oplus S_0.
\]

Consider the restriction map $r_M$ to $M$.  We claim that the restriction of
$S_M \oplus S_0$ to $M$ is injective:
\[
  \dim\bigl( S_M \oplus r_M(S_0) \bigr) = \dim(S_M) + \dim(S_0).
\]
Indeed, if a nontrivial linear combination of elements of $S_M \oplus S_0$ restricts to zero
on $M$, then it is supported on $M^c$, hence lies in $S \cap P_{M^c} = S_{M^c}$, which is
impossible unless the combination was trivial (since $S_M \oplus S_{M^c} \oplus S_0$ is a
direct sum).

Now compute $\dim(S^\perp \cap P_M)$.  Since $P_M$ is symplectically orthogonal to $S_{M^c}$
(their supports are disjoint), the constraints for a vector in $P_M$ to commute with all of
$S$ are exactly orthogonality to $r_M(S_M \oplus S_0) = S_M \oplus r_M(S_0)$.  Therefore,
\[
  \dim(S^\perp \cap P_M) = \dim(P_M) - \dim(S_M) - \dim(S_0) = 2|M| - \dim(S_M) - \dim(S_0).
\]
Similarly,
\[
  \dim(S^\perp \cap P_{M^c}) = 2|M^c| - \dim(S_{M^c}) - \dim(S_0).
\]

By definition,
\[
  g(M) = 2|M| - 2\dim(S_M) - \dim(S_0),
  \qquad
  g(M^c) = 2|M^c| - 2\dim(S_{M^c}) - \dim(S_0).
\]
Adding and using $|M| + |M^c| = n$ and
$\dim(S) = \dim(S_M) + \dim(S_{M^c}) + \dim(S_0) = n - k$ gives
\[
  g(M) + g(M^c) = 2n - 2\dim(S) = 2k. \qedhere
\]
\end{proof}

\begin{remark}
The identity $g(M) + g(M^c) = 2k$ does not imply that each logical Pauli operator is
supportable on exactly one of $M$ or $M^c$.  It is a statement about dimensions, not about
individual operators.
\end{remark}

\begin{corollary}[Cleaning lemma --- support on the complement]\label{cor:cleaning}
If $M$ is correctable, then every logical class $[L] \in S^\perp / S$ has a representative
supported on $M^c$.
\end{corollary}

\begin{proof}
If $M$ is correctable, then $S^\perp \cap V_M \subseteq S$, hence $g(M) = 0$.  By the
dimension identity, $g(M^c) = 2k = \dim(S^\perp / S)$.
\end{proof}

\section{Two disjoint correctable sets}\label{sec:two-disjoint}

The following lemma is the key step connecting the cleaning lemma to the Singleton bound.  It
shows that if we can find two disjoint correctable sets, then the number of logical qubits is
bounded by the size of the remaining positions.

\begin{lemma}[Two disjoint correctable sets bound logical dimension]\label{lem:two-disjoint}
Partition the physical qudits into $A\, B\, C$ with $A \cap B = \varnothing$ and
$C = [n] \setminus (A \cup B)$.  If erasures of both $A$ and $B$ are correctable, then
\[
  k \le |C|.
\]
\end{lemma}

\begin{proof}
Set $D := A^c = B \cup C$.  Since $A$ is correctable we have $g(A) = 0$, hence by the
cleaning dimension identity (Lemma~\ref{lem:cleaning}),
\[
  g(D) = g(A^c) = 2k.
\]
By definition,
\[
  g(D) = \dim\Bigl( (S^\perp \cap V_D) / (S \cap V_D) \Bigr),
\]
so the quotient space $L := (S^\perp \cap V_D) / (S \cap V_D)$ has dimension $\dim(L) = 2k$.

Consider the restriction map $r_C : V_D \to V_C$ (zeroing coordinates outside $C$).  Let
$W := S \cap V_D$.  Then $r_C(W) \le V_C$ is a subspace, and $r_C$ induces a linear map
\[
  \bar{r}_C : L \longrightarrow V_C / r_C(W),
  \qquad [v] \longmapsto r_C(v) \bmod r_C(W).
\]
This is well-defined because replacing $v$ by $v + w$ with $w \in W$ changes $r_C(v)$ by an
element of $r_C(W)$.

We claim $\bar{r}_C$ is injective.  Suppose $\bar{r}_C([v]) = 0$.  Then
$r_C(v) \in r_C(W)$, so there exists $w \in W$ with $r_C(v) = r_C(w)$.  Hence
$r_C(v - w) = 0$, so $v - w$ is supported entirely on $B$, i.e.,
$v - w \in V_B$.  Also $v - w \in S^\perp$ (since $v \in S^\perp$ and $w \in S$).  Thus
$v - w \in S^\perp \cap V_B$, and because $B$ is correctable we get $v - w \in S$.  Since
$w \in S$ as well, it follows that $v \in S$, hence $[v] = 0$ in $L$.

Therefore,
\[
  2k = \dim(L) \le \dim\bigl( V_C / r_C(W) \bigr) \le \dim(V_C) = 2|C|,
\]
so $k \le |C|$.
\end{proof}

\section{The Quantum Singleton Bound}\label{sec:singleton}

\begin{theorem}[Quantum Singleton Bound]\label{thm:QSB}
For any $[[n, k, d]]$ stabiliser code,
\[
  k + 2(d-1) \le n.
\]
\end{theorem}

\begin{proof}
Choose disjoint subsets $A, B \subseteq [n]$ with $|A| = |B| = d - 1$ and let
$C = [n] \setminus (A \cup B)$.  By Lemma~\ref{lem:dist-correct}, erasures of $A$ and $B$
are both correctable.  Lemma~\ref{lem:two-disjoint} then yields
\[
  k \le |C| = n - 2(d-1),
\]
which is equivalent to the stated bound.
\end{proof}

\begin{remark}[Existence of disjoint correctable sets]
In the proof above we choose disjoint subsets $A, B \subseteq [n]$ with $|A| = |B| = d - 1$.
Such a choice is possible whenever $n \ge 2(d-1)$.

If instead $n < 2(d-1)$, then no such disjoint subsets exist.  In this regime, the inequality
$k \le n - 2(d-1)$ has a non-positive right-hand side, and since $k \ge 0$ for any stabiliser
code, the bound is automatically satisfied.  Hence the argument above covers all cases.
\end{remark}

\section{Discussion}\label{sec:discussion}

We have derived the Quantum Singleton Bound from three ingredients: (i) the symplectic vector
space model of stabiliser codes, (ii) distance-based erasure correctability, and (iii) a
dimension-counting proof of the cleaning lemma following
Preskill~\cite{Preskill2022TopologicalStabilizerCodes}.  The final inequality follows from a
purely symplectic dimension argument using the correctability of two disjoint erasures.

\paragraph{Comparison with entropic proofs.}
The entropy-based proof of the Quantum Singleton
Bound~\cite{NielsenChuang,GrasslHuberWinter2022} proceeds by noting that if a code corrects
the erasure of a set~$E$, then the encoded information is fully determined by the
complementary qudits, which yields $S(A) \ge S(AB)$ for appropriate subsystem decompositions
(where $S$ denotes von~Neumann entropy).  The entropy approach is more general: it applies to
arbitrary quantum codes, not only stabiliser codes, and it naturally handles approximate error
correction and noisy settings.  However, for the stabiliser case, the additional generality is
not needed, and the symplectic proof makes the algebraic mechanism transparent.

\paragraph{Scope and limitations.}
The proof presented here applies to stabiliser codes over prime fields~$\F_p$.  The extension
to stabiliser codes over non-prime fields $\F_q = \F_{p^m}$ (using Hermitian or
trace-symplectic forms) is straightforward and follows the same outline, as the relevant
dimension identities remain valid in that
setting~\cite{KetkarKlappeneckerKumarSarvepalli2006,Ashikhmin2001}.  For general (possibly
non-additive) quantum codes, the entropic argument remains the most natural route.

\paragraph{The formalisation.}
The \texttt{Lean4} formalisation required several auxiliary developments not present in the
current \texttt{Mathlib} library, including the symplectic form as a bilinear form, properties
of support subspaces and restriction maps, and the interaction between symplectic orthogonal
complements and subspace intersections.  The most involved part of the formalisation was the
cleaning dimension identity (Lemma~\ref{lem:cleaning}), which required careful bookkeeping of
the direct-sum decomposition $S = S_M \oplus S_{M^c} \oplus S_0$.  The final theorem
(Theorem~\ref{thm:QSB}) itself was a relatively short consequence once the lemmas were in
place.  We hope that the auxiliary material developed here will be useful for future
formalisation efforts in quantum coding theory.


\bibliography{main}
\bibliographystyle{alpha}


\appendix
\section{Lean4 formalisation and code index}\label{app:lean}

\subsection{Repository, version, and permalinks}
The formalisation lives in the \texttt{tcslib} repository, in the file
\href{https://github.com/Shilun-Allan-Li/tcslib/blob/9c7f65d5b18b64bc56760a2a09dfe6fe3305b8d6/TCSlib/ErrorCorrectingCodes/QuantumSingleton.lean}{\texttt{TCSlib/ErrorCorrectingCodes/QuantumSingleton.lean}}.
All permalinks in the tables below are pinned to commit
\href{https://github.com/Shilun-Allan-Li/tcslib/tree/9c7f65d5b18b64bc56760a2a09dfe6fe3305b8d6}{\texttt{9c7f65d}},
which compiles against \texttt{Mathlib} revision \texttt{cd0d357} on the Lean toolchain
\texttt{leanprover/lean4:v4.25.0-rc2}.  The main file contains no \texttt{sorry}, \texttt{admit},
or introduced \texttt{axiom}s.

\subsection{Pointers to core definitions}
\begin{center}
\renewcommand{\arraystretch}{1.15}
\begin{tabular}{|p{0.54\textwidth}|p{0.40\textwidth}|}
\hline
\textbf{Mathematical object (paper)} & \textbf{Lean declaration (GitHub link)}\\
\hline
Prime field $\F=\F_p$ &
\href{https://github.com/Shilun-Allan-Li/tcslib/blob/9c7f65d5b18b64bc56760a2a09dfe6fe3305b8d6/TCSlib/ErrorCorrectingCodes/QuantumSingleton.lean#L72-L74}{\texttt{abbrev F}}\\
\hline
Ambient space $V\cong \F^{2n}$ &
\href{https://github.com/Shilun-Allan-Li/tcslib/blob/9c7f65d5b18b64bc56760a2a09dfe6fe3305b8d6/TCSlib/ErrorCorrectingCodes/QuantumSingleton.lean#L75-L77}{\texttt{abbrev V}}\\
\hline
Standard symplectic form $\langle\cdot,\cdot\rangle$ &
\href{https://github.com/Shilun-Allan-Li/tcslib/blob/9c7f65d5b18b64bc56760a2a09dfe6fe3305b8d6/TCSlib/ErrorCorrectingCodes/QuantumSingleton.lean#L78-L79}{\texttt{def sym\_form}}\\
\hline
Bilinear form packaging (as \texttt{LinearMap.BilinForm}) &
\href{https://github.com/Shilun-Allan-Li/tcslib/blob/9c7f65d5b18b64bc56760a2a09dfe6fe3305b8d6/TCSlib/ErrorCorrectingCodes/QuantumSingleton.lean#L190-L205}{\texttt{noncomputable def symB}}\\
\hline
Support and weight $\supp(\cdot),\ \wt(\cdot)$ &
\href{https://github.com/Shilun-Allan-Li/tcslib/blob/9c7f65d5b18b64bc56760a2a09dfe6fe3305b8d6/TCSlib/ErrorCorrectingCodes/QuantumSingleton.lean#L220-L226}{\texttt{def supp}, \texttt{def wt}}\\
\hline
Support subspace $V_C$ (vectors supported on $C$) &
\href{https://github.com/Shilun-Allan-Li/tcslib/blob/9c7f65d5b18b64bc56760a2a09dfe6fe3305b8d6/TCSlib/ErrorCorrectingCodes/QuantumSingleton.lean#L227-L243}{\texttt{def V\_sub}}\\
\hline
Dimension of a support subspace $\dim(V_C)=2|C|$ &
\href{https://github.com/Shilun-Allan-Li/tcslib/blob/9c7f65d5b18b64bc56760a2a09dfe6fe3305b8d6/TCSlib/ErrorCorrectingCodes/QuantumSingleton.lean#L316-L328}{\texttt{lemma dim\_V\_sub}}\\
\hline
Restriction map $r_E$ &
\href{https://github.com/Shilun-Allan-Li/tcslib/blob/9c7f65d5b18b64bc56760a2a09dfe6fe3305b8d6/TCSlib/ErrorCorrectingCodes/QuantumSingleton.lean#L329-L365}{\texttt{def r\_E}}\\
\hline
Symplectic orthogonal complement $S^\perp$ &
\href{https://github.com/Shilun-Allan-Li/tcslib/blob/9c7f65d5b18b64bc56760a2a09dfe6fe3305b8d6/TCSlib/ErrorCorrectingCodes/QuantumSingleton.lean#L366-L370}{\texttt{abbrev sym\_orth}}\\
\hline
Isotropic subspace predicate &
\href{https://github.com/Shilun-Allan-Li/tcslib/blob/9c7f65d5b18b64bc56760a2a09dfe6fe3305b8d6/TCSlib/ErrorCorrectingCodes/QuantumSingleton.lean#L371-L374}{\texttt{def IsIsotropic}}\\
\hline
Code distance (as an infimum over weights) &
\href{https://github.com/Shilun-Allan-Li/tcslib/blob/9c7f65d5b18b64bc56760a2a09dfe6fe3305b8d6/TCSlib/ErrorCorrectingCodes/QuantumSingleton.lean#L375-L383}{\texttt{noncomputable def code\_dist}}\\
\hline
Correctable erasure (commutant form) &
\href{https://github.com/Shilun-Allan-Li/tcslib/blob/9c7f65d5b18b64bc56760a2a09dfe6fe3305b8d6/TCSlib/ErrorCorrectingCodes/QuantumSingleton.lean#L384-L392}{\texttt{def correctable}}\\
\hline
Supportable logical operator count $g(M)$ &
\href{https://github.com/Shilun-Allan-Li/tcslib/blob/9c7f65d5b18b64bc56760a2a09dfe6fe3305b8d6/TCSlib/ErrorCorrectingCodes/QuantumSingleton.lean#L421-L437}{\texttt{noncomputable def g}}\\
\hline
Logical dimension parameter $k = n-\dim(S)$ &
\href{https://github.com/Shilun-Allan-Li/tcslib/blob/9c7f65d5b18b64bc56760a2a09dfe6fe3305b8d6/TCSlib/ErrorCorrectingCodes/QuantumSingleton.lean#L994-L1001}{\texttt{def code\_k}}\\
\hline
\end{tabular}
\end{center}

\subsection{Pointers to key lemmas and the main theorem}
\begin{center}
\renewcommand{\arraystretch}{1.15}
\begin{tabular}{|p{0.54\textwidth}|p{0.40\textwidth}|}
\hline
\textbf{Result (paper)} & \textbf{Lean lemma/theorem (GitHub link)}\\
\hline
Nondegeneracy of the symplectic form &
\href{https://github.com/Shilun-Allan-Li/tcslib/blob/9c7f65d5b18b64bc56760a2a09dfe6fe3305b8d6/TCSlib/ErrorCorrectingCodes/QuantumSingleton.lean#L206-L219}{\texttt{lemma sym\_form\_nondegenerate}}\\
\hline
Distance implies erasure correctability &
\href{https://github.com/Shilun-Allan-Li/tcslib/blob/9c7f65d5b18b64bc56760a2a09dfe6fe3305b8d6/TCSlib/ErrorCorrectingCodes/QuantumSingleton.lean#L393-L412}{\texttt{lemma dist\_implies\_correctable}}\\
\hline
Cleaning dimension identity (formal analogue of $g(M)+g(M^c)=2k$) &
\href{https://github.com/Shilun-Allan-Li/tcslib/blob/9c7f65d5b18b64bc56760a2a09dfe6fe3305b8d6/TCSlib/ErrorCorrectingCodes/QuantumSingleton.lean#L869-L882}{\texttt{lemma cleaning\_dimension\_identity}}\\
\hline
Two disjoint correctable sets bound logical dimension ($k\le |C|$) &
\href{https://github.com/Shilun-Allan-Li/tcslib/blob/9c7f65d5b18b64bc56760a2a09dfe6fe3305b8d6/TCSlib/ErrorCorrectingCodes/QuantumSingleton.lean#L1002-L1031}{\texttt{lemma two\_disjoint\_correctable\_sets
\_bound\_logical\_dimension}}\\
\hline
Existence of disjoint sets of prescribed size (helper for the final step) &
\href{https://github.com/Shilun-Allan-Li/tcslib/blob/9c7f65d5b18b64bc56760a2a09dfe6fe3305b8d6/TCSlib/ErrorCorrectingCodes/QuantumSingleton.lean#L1218-L1252}{\texttt{lemma exists\_disjoint\_finsets\_card}}\\
\hline
Quantum Singleton bound &
\href{https://github.com/Shilun-Allan-Li/tcslib/blob/9c7f65d5b18b64bc56760a2a09dfe6fe3305b8d6/TCSlib/ErrorCorrectingCodes/QuantumSingleton.lean#L1253-L1399}{\texttt{theorem quantum\_singleton\_bound}}\\
\hline
\end{tabular}
\end{center}

\end{document}